\documentclass[%
 reprint,
superscriptaddress,
amsmath,amssymb,
aps,
floatfix,
]{revtex4-2}

\usepackage{graphicx} 
\usepackage{amsmath} 
\usepackage{amssymb}  
\usepackage{bbold}
\usepackage{color}
\usepackage{amsthm}

\DeclareMathOperator{\diag}{diag}
\newtheorem{theorem}{Theorem}
\newtheorem{lemma}[theorem]{Lemma}
\newtheorem{definition}{Definition}

\begin{document}

\preprint{APS/123-QED}

\title{Bounding first-order quantum phase transitions in adiabatic quantum computing}

\author{Matthias Werner}
\email{matthias.werner@qilimanjaro.tech}
\affiliation{Qilimanjaro Quantum Tech., Carrer dels Comtes de Bell-Lloc, 161, 08014 Barcelona, Spain}
\affiliation{Departament de Física Quàntica i Astrofísica (FQA), Universitat de Barcelona (UB), Carrer de Martí i Franqués, 1, 08028 Barcelona, Spain}

\author{Artur García-Sáez}
\affiliation{Qilimanjaro Quantum Tech., Carrer dels Comtes de Bell-Lloc, 161, 08014 Barcelona, Spain}
\affiliation{Barcelona Supercomputing Center, Plaça Eusebi Güell, 1-3, 08034 Barcelona, Spain}

\author{Marta P. Estarellas}
\affiliation{Qilimanjaro Quantum Tech., Carrer dels Comtes de Bell-Lloc, 161, 08014 Barcelona, Spain}

\thispagestyle{empty}
\pagestyle{empty}

\begin{abstract}
    In the context of adiabatic quantum computation (AQC), it has been argued that first-order quantum phase transitions (QPTs) due to localisation phenomena cause AQC to fail by exponentially decreasing the minimal spectral gap of the Hamiltonian along the annealing path as a function of the qubit number. The vanishing of the spectral gap is often linked to the localisation of the ground state in a local minimum, requiring the system to tunnel into the global minimum at a later stage of the annealing. Recent methods have been proposed to avoid this phenomenon by carefully designing the involved Hamiltonians. However, it remains a challenge to formulate a comprehensive theory of the effect of the various parameters and the conditions under which QPTs make the AQC algorithm fail. Equipped with concepts from graph theory, in this work we link graph quantities associated to the Hamiltonians along the annealing path with the occurrence of QPTs. These links allow us to derive bounds on the location of the minimal spectral gap along the annealing path, augmenting the toolbox for the analysis of strategies to improve the runtime of AQC algorithms.
\end{abstract}

\maketitle

\section{INTRODUCTION}
One of the central goals of quantum computing is the prospect of being able to efficiently solve classically hard computational problems. Adiabatic Quantum Computation (AQC), proposed by Farhi et al. \cite{Farhi_2000}, is a model of quantum computation particularly well suited to tackle optimization tasks that fall into this category. Roland et al. \cite{Roland_2002} showed that a quadratic speed-up of Grover's search algorithm \cite{Grover_1996} can be obtained not only through a gate-based quantum circuit but also by AQC. This, together with the proofs of equivalence between AQC and the gate-based model \cite{Aharonov2004}, indicate that a universal AQC device would provide quantum advantage.\\ 
In AQC, a quantum system is prepared in the ground state of a relatively simple initial Hamiltonian, also called the driver Hamiltonian. The Hamiltonian of the system is then \textit{slowly} interpolated to a target Hamiltonian whose ground state encodes the solution of the target problem. By \textit{slowly} we mean that the rate of change of the Hamiltonian adheres to the adiabatic condition as stated by the adiabatic theorem \cite{Albash2018}. As a consequence, the runtime of the algorithm is inversely related to the width of the spectral gap of the instantaneous Hamiltonian and a rapidly closing spectral gap therefore dramatically increases the runtime, making the algorithm infeasible. A major cause of these exploding runtimes are first-order quantum phase transitions (QPTs) \cite{Young_2010, Thomas_2010} due to Anderson localisation, which results in (avoided) level crossings that lead to an exponential closing of the spectral gap and, consequently, to an exponential runtime in the number of qubits in the system. Altshuler et al. \cite{Altshuler_2010} considered this to be a failure proof of AQC. However, this proposition has been contested, as methods are known to avoid the exponential closing of the gap \cite{Knysh_2010, Choi_2011, Farhi_2011, Dickson_2011a, Dickson_2011b}, suggesting that there are specific conditions when localisation phenomena can be avoided by careful design of the initial Hamiltonian.\\
In the context of AQC, first-order QPTs can occur when an initially delocalized state transitions into a localized state that is supported in a local minimum, while only having negligible amplitudes in the global minimum. As a consequence, as the annealing continues, the ground state transitions to the global minimum, resulting in a rapidly closing spectral gap as well as a discontinuity in the solution fidelity. The latter transition from the local to the global minimum constitutes a first-order QPT. In the spectrum of the interpolated Hamiltonians the first-order QPTs correspond to (avoided) level crossings of the ground and first excited state. However, ideally these transitions are avoided and the delocalized state transitions directly into the global minimum, which results in a smoother fidelity profile. The two scenarios are depicted in Figure \ref{figLocalizationLandscape}. Such a qualitative difference raises the question of what are the distinguishing properties of the local minima that make the ground state localize there first. Amin and Choi linked the occurrence of first-order QPTs to the presence of a large number of low-energy states which are connected by a small number of bit flips \cite{AminChoi_2009}. This notion, however, is rather broad.\\
In this work we push towards answering this question from a graph-theoretical perspective. By applying degenerate perturbation theory we show how particular graph theoretic quantities obtained from the initial Hamiltonian can be related to the spectral gap along the annealing process. These quantities are well understood in spectral graph theory, linking them to the spectral gap of adjacency matrices \cite{AlShimary2010, Crosson2017, Jarret2018_1, Jarret2018_2}. Importantly, these links allow us to give conditions for the occurrence or absence of first-order QPTs and derive bounds on its location, shedding some light on this particular error mechanism common in AQC algorithms with the prospect of finding strategies to mitigate it. The use of graph theory has proven to be a useful tool in the understanding of many-body systems \cite{Estarellas2020,Bastidas2018}.\\
This work is structured as follows: first we review the basics of the AQC model in section \ref{secAQC}. In the following section \ref{secTheory} we give the necessary definitions and show how degenerate perturbation theory allows for the introduction of certain graph theoretical concepts, specifically the conductance of a subset of nodes $V$ and the maximum degree of the respective induced sub-graph $G(V)$. Using these concepts, we derive bounds on the energy of states localized in local minima of the energy landscape, which further allows us to derive bounds on the location of the minimal spectral gap along the annealing path. In section \ref{secNumerics}, in order to numerically investigate the validity of the derived bounds, we then exactly solve artificially generated toy model instances, as well as an instance of an NP-complete problem and compare the observed location of the minimal spectral gap with the predictions of our bounds. We conclude this work by discussing the tightness and interpretation of the derived bounds in section \ref{secDiscussion}, as well as potential applications of our results.

\begin{figure}
    \centering
    \includegraphics[scale=0.5]{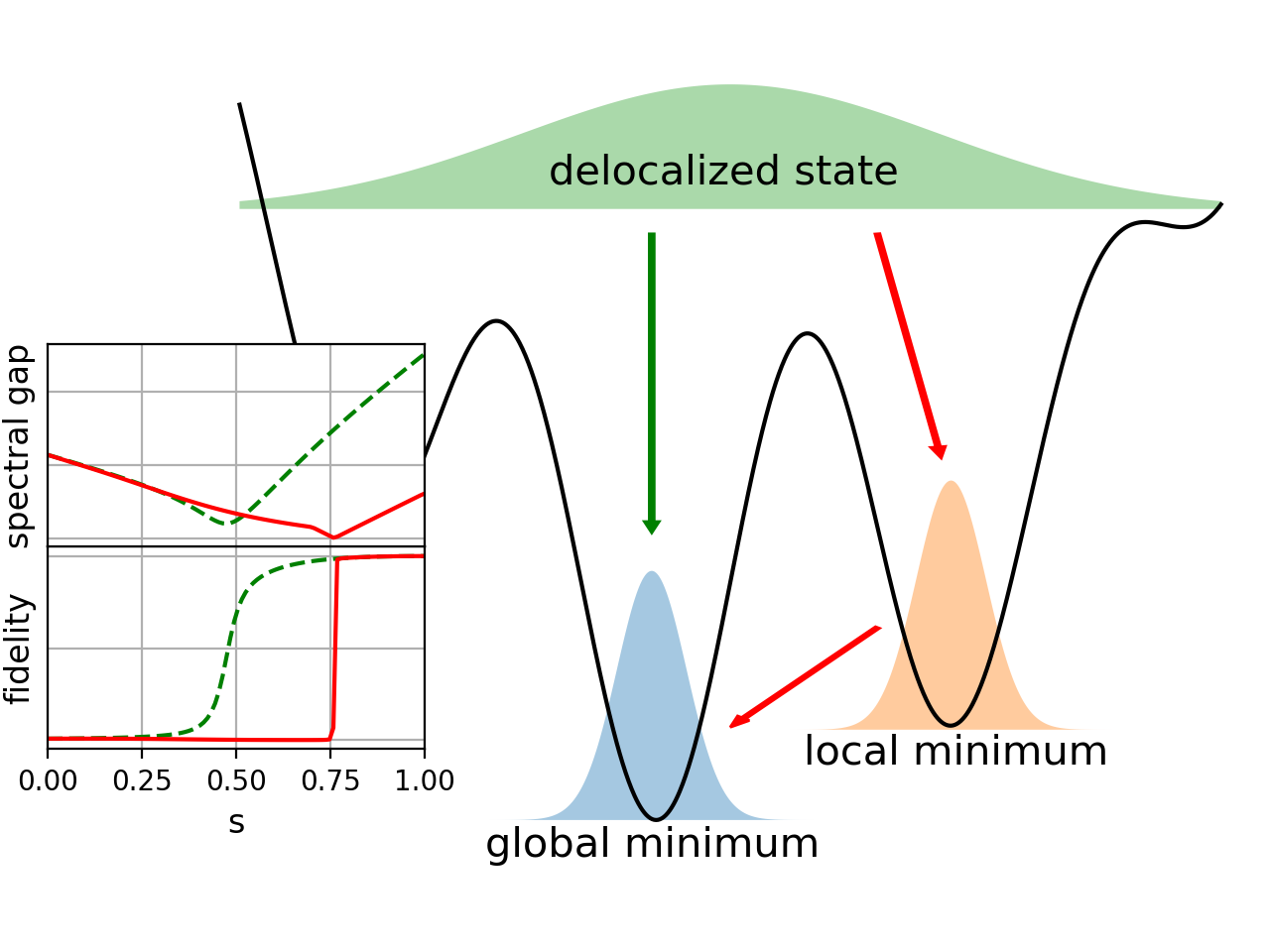}
    \caption{Ground state localisation with (red arrows and solid graph profiles) and without (green arrow and dashed graph profiles) tunneling to the global minimum; the inset shows the spectral gaps and solution fidelities over the interpolation parameter, or annealing schedule, $s$. }
    \label{figLocalizationLandscape}
\end{figure}

\section{Adiabatic Quantum Computation}
\label{secAQC}
AQC works by interpolating between a driver Hamiltonian $H_D$, also called initial Hamiltonian, and the target Hamiltonian $H_T$. The ground state of $H_D$ needs to be simple to prepare, while $H_T$ has been carefully designed such that the ground state encodes the solution of the problem at hand. In the case of optimization problems, this is often done by formulating the problem to a quadratic unconstrained binary optimization (QUBO) \cite{Lucas_2014}.\\
The full Hamiltonian as a function of the interpolation parameter $s = s(t) \in [0, 1]$, also called the schedule, is given by
\begin{equation}
    H(s) = (1-s)H_D + s H_T \ .
    \label{eqAdiabaticTransitionHamiltonian}
\end{equation}
The instantaneous eigenstates of $H(s)$ are denoted by $|\Psi_n(s)\rangle$ with respective eigenvalues $E_n = E_n(s)$ for $n = 0, ..., N-1$ such that $E_0 \leq E_1 \leq ... \leq E_{N-1}$ with $N = 2^{N_Q}$ the dimension of the Hilbert space and $N_Q$ the number of qubits. Let us consider
\begin{equation}
    |\Psi(s)\rangle = \sum_n a_n(s) |\Psi_n(s)\rangle
\end{equation} an arbitrary state of the quantum system along the anneal path. At $s=0$ we prepare the system such that $|a_0(s=0)|^2 = 1$. In order to ensure $|a_0(s=1)|^2 \approx 1$, the rate of change has to obey the adiabatic theorem \cite{Albash2018}
\begin{equation}
    \frac{|\langle \Psi_1(s) |\frac{dH}{ds}|\Psi_0(s)\rangle|}{g_{\mathrm{min}}^2}\leq \epsilon 
    \label{eqAdiabaticCondition}
\end{equation}
for $\epsilon \ll 1$ and where
\begin{equation}
    g_{\mathrm{min}} = \min_{s} \left( E_1(s) - E_0(s) \right)
\end{equation}
is the minimal spectral gap between ground and first excited state. The matrix element in the numerator of Eq. \eqref{eqAdiabaticCondition} can typically be assumed to be bounded by a polynomial of the qubit number, however in this work we normalize the Hamiltonians such that the matrix element is bounded by a constant. This serves to simplify some of the expressions and isolate the effect of the underlying graph structure of $H_D$, as discussed below. Given this normalization, the runtime of an AQC algorithm is determined by the minimal gap $g_{\mathrm{min}}$.\\
In this work we will assume the target Hamiltonians $H_T$ to be diagonal in the computational basis, i.e.
\begin{equation}
    H_T = \diag(E_0^T, E_1^T, ..., E_{N-1}^T)
    \label{eqTarget}
\end{equation}
with eigenstates $|z\rangle$ for each eigenvalue $E_z^T$. 
A common choice for $H_D$ is
\begin{equation}
    H_D = -\sum_{i=0}^{N_Q-1} \sigma_i^x
    \label{eqSigmaXDriver}
\end{equation}
where $\sigma_i^x$ is the Pauli-x operator applied on qubit $i$. Various works \cite{Hen_2016, Choi_2021} investigated the impact of $H_D$ on the spectral gap and hence on the runtime. In this work we investigate the impact of $H_D$ on the runtime as well. However, we will focus on the underlying graph of the Hamiltonian and its relation to the occurrence of first-order QPTs.

\section{Graph theory and QPT\lowercase{s}}
\label{secTheory}
\subsection{Basic definitions}
We consider driver Hamiltonians $H_D$ that can be associated with a graph $G$ in the Hilbert space. A graph $G = (\mathcal{V}, \mathcal{E})$ is defined by a set of nodes $\mathcal{V}$, as well as a set of edges
\begin{equation}
    \mathcal{E} := \{ (i,j) : i,j \in \mathcal{V} \text{ connected in } G\} \ .
\end{equation}
Figure \ref{figGridBasis} shows an example of a graph. For each graph $G$ one can define the adjacency matrix $A_G \in \{0, 1\}^{|\mathcal{V}| \times |\mathcal{V}|}$ with
\begin{equation}
    (A_G)_{ij} = \begin{cases}
            1 \text{ if } (i,j) \in \mathcal{E}, \\
            0 \text{ else \ .}
        \end{cases}
\end{equation}
As it is the case with the common driver Hamiltonian defined in Eq. \eqref{eqSigmaXDriver}, we assume the elements of $H_D$ to be negative. Then we can write more generally
\begin{equation}
    H_D = \frac{-1}{d} A_G
    \label{eqHdDefinition}
\end{equation}
with the adjacency matrix of a $d$-regular simple graph $G=(\mathcal{V}, \mathcal{E})$ where the nodes $\mathcal{V}$ are the computational basis states, denoted by $|z\rangle$, and the edges $\mathcal{E}$ are the non-zero matrix elements of $H_D$ as depicted in Figure \ref{figGridBasis}.\\
For simplicity, we will limit the analysis to $d$-regular simple graphs, as many commonly used $H_D$ such as Eq. \eqref{eqSigmaXDriver} fall into this category. The scaling by $\frac{1}{d}$ is introduced to normalize the ground state energies of the investigated $H_D$ to -1.
\begin{figure}
    \centering
    \includegraphics[scale=0.35]{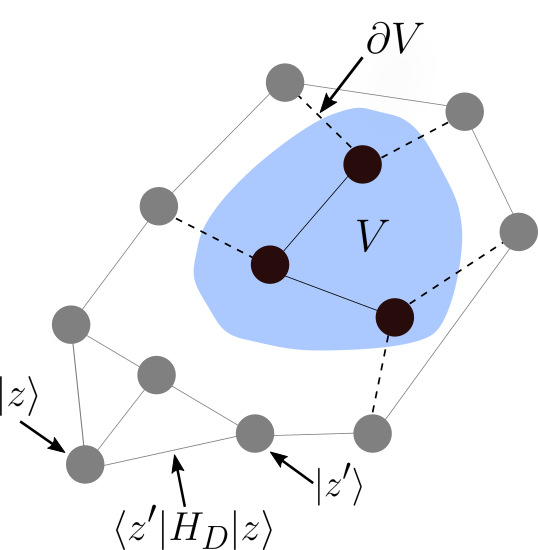}
    \caption{An example of a 3-regular Graph $G$ of the configuration space spanned by the eigenvectors $|z \rangle$ of $H_T$. The $|z \rangle$ are represented by the nodes $\mathcal{V}$, while the set of edges $\mathcal{E}$ correspond to the off-diagonals given by the matrix elements of $H_D$. The nodes in the shaded area represent the degenerate subspace $V$ of $H_T$, inducing the subgraph $G(V)$ (black nodes, solid black edges). All edges leaving $V$ (dashed black edges) constitute the edge boundary $\partial V$ of $V$.}
    \label{figGridBasis}
\end{figure}
Moreover, we will make use of the following concepts:
\begin{definition}(Induced subgraph)
Let $G = (\mathcal{V}$, $\mathcal{E})$ be a graph and $V \subseteq \mathcal{V}$. The induced subgraph $G(V) \subseteq G$ is defined as the graph
\begin{equation}
    G(V) = (V, E)
\end{equation}
with
\begin{equation}
    E = \{(i,j) \in \mathcal{E} : i,j \in V \} \ .
\end{equation}
\end{definition}

\begin{definition}(Edge boundary)
Let G = ($\mathcal{V}$, $\mathcal{E}$) a graph and $V \subseteq \mathcal{V}$. The edge boundary $\partial V \subseteq \mathcal{E}$ of $V$ is defined as
\begin{equation}
    \partial V = \{ (i, j) \in \mathcal{E}: i \in V, j \in \mathcal{V} \backslash V \} \ .
\label{eqEdgeBoundaryDefinition}
\end{equation}
\end{definition}

\begin{definition}(Conductance)
Let G = ($\mathcal{V}$, $\mathcal{E}$) a graph and $V \subseteq \mathcal{V}$. The conductance $\phi (V)$ of $V$ is defined as
\begin{equation}
    \phi(V) = \frac{|\partial V|}{|V|} \ .
\label{eqConductanceDefinition}
\end{equation}
\end{definition}

In a slight abuse of notation, we will use the symbol $V$ for both the subset of nodes in $G$, as well as the subspace of the Hilbert space spanned by (nearly) degenerate eigenstates of $H_T$.

\subsection{From degenerate perturbation theory to spectral graph theory}
\label{secTheorySpecTheo}
We will make use of degenerate perturbation theory. To this end let us define the set $V$ as the set of (nearly) degenerate eigenstates of $H_T$ with $E_z^T \approx E^T_V$, with $E^T_V$ being the energy of the local minimum $V$. Considering $sH_T$ the unperturbed Hamiltonian and $(1-s)H_D$ the perturbation, we have to diagonalize $H(s)$ on the subspace spanned by $V$. Given that $H_T$ is (nearly) degenerate in this subspace, we have to solve the eigenvalue equation
\begin{equation}
    E_{V}(s) | V \rangle = (1-s) H_D' | V \rangle + s E_V^T | V \rangle
    \label{eqSubspaceEigenvalueEquation}
\end{equation}
where $H_D'$ is the projection of $H_D$ onto the subspace $V$. Note that $|V \rangle$ by definition is an element of the subspace $V$ and hence we consider it a state localized in $V$ with energy $E_V(s)$.\\
From Eq. \eqref{eqSubspaceEigenvalueEquation} it follows directly that $|V\rangle$ has to be an eigenvector of
\begin{equation}
    H_D' = \frac{-1}{d} A_{G(V)}    
\end{equation}
with $G(V)$ the sub-graph of $G$ induced by $V$ and $A_{G(V)}$ its adjacency matrix. The eigenvector with the minimal energy is the principal eigenvector of $G(V)$ and its energy is given by
\begin{equation}
    E_V(s) = -(1-s) \frac{1}{d} \lambda_V + sE_V^T
\end{equation}
where $\lambda_V$ is the principal eigenvalue of $G(V)$. If the subgraph $G(V)$ is connected, the principal eigenvalue is unique and the perturbation with the driver Hamiltonian $H_D$ lifts the degeneracy of $H_T$. However, our analysis remains valid for disconnected $G(V)$ as well. If $G(V)$ is disconnected, one can consider each connected components in $G(V)$ separately, as in this case the set $V$ of degenerate eigenstates of $H_T$ corresponds to distinct local minima in $G$. The principal eigenvalue $\lambda_V$ is going to be the largest principal eigenvalue of the induced subgraphs of the connected components. Equivalently, one can consider as local minima only subsets $V$ such that $G(V)$ is connected.\\
It can be shown for a $d$-regular graph $G$ that
\begin{equation}
    d - \phi(V) \leq \lambda_V \leq d_{\mathrm{max}}(V)
    \label{eqPrincipalEigenvalueBounds}
\end{equation}
where $d_{\mathrm{max}}(V)$ is the maximal degree of $G(V)$ (see Appendix \ref{secAppendixDerivationOfBounds}). Note that while $G$ is $d$-regular, $G(V)$ may be irregular and $d_{\mathrm{max}}(V) \leq d$. Using these ingredients we obtain both a lower and an upper bound on $E_{V}(s)$
\begin{subequations}
\begin{align}
    E_{V}(s) &\geq -(1-s) \frac{d_{\mathrm{max}}(V)}{d} + s E_V^T \ , \\[10pt]
    E_{V}(s) &\leq (1-s) \left( \frac{\phi(V)}{d} - 1 \right) + s E_V^T \ . \label{eqEVupperbound}
\end{align}
\label{eqEVbound}%
\end{subequations}
The bounds Eq. \eqref{eqEVbound} are the first key result of this work. In the next section, we will use them to derive bounds on the location of first-order quantum phase transitions along the anneal.

\subsection{Bounding first-order quantum phase transitions}
QPTs \cite{Vojta2002, Sachdev} are also called zero-temperature phase transitions, as they are driven by the competition between quantum fluctuations and minimizing some potential. Classical phase transitions are driven instead by entropic fluctuations. At zero temperature, the entropic part of the potential goes to zero, but quantum fluctuations persist. Level crossings of the ground and first excited states in that case can be seen as first-order phase transitions according to the Ehrenfest classification, since in thermal equilibrium at zero temperature only the ground state is populated. Consequently, the thermodynamic free energy at zero temperature is equal to the ground state energy. The first derivative of the free energy with respect to the annealing schedule $s$ will be discontinuous at the level crossing, since the ground state changes rapidly and
\begin{equation}
    \begin{aligned}
    \frac{d}{d s} &\langle \psi_0 | H(s) | \psi_0 \rangle \\
    = &\langle \psi_0 | \frac{d}{d s} H(s) | \psi_0 \rangle \\
    = &\langle \psi_0 | H_T - H_D | \psi_0 \rangle
    \end{aligned}
\end{equation}
Therefore, level crossings can be considered a first-order QPT \cite{AminChoi_2009, TakadaLidar}.\\
Using the bounds  Eq. \eqref{eqEVbound}, it is possible to estimate the location of the crossing of two energy levels within first-order perturbation theory. There are two conditions to be met for a level crossing to occur. Let $E_{\mathrm{local}}(s)$ and $E_{\mathrm{global}}(s)$ the energies of states localized in the potentially degenerate local and global minima respectively, as depicted in Figure \ref{figLevelCrossingPrinciple}. First, the energies are required to cross at some value of $s^* \in[0, 1]$
\begin{equation}
E_{\mathrm{global}}(s^*) = E_{\mathrm{local}}(s^*) \ .
\label{eqLocLocCondition}
\end{equation}
However, this is not sufficient for the level crossing to lead to a first-order QPT. Consider $E_{\mathrm{deloc}}(s)$ to be the energy of a delocalized state. If at $s^*$ we find
\begin{equation}
    E_{\mathrm{deloc}}(s^*) < E_{\mathrm{local}}(s^*) = E_{\mathrm{global}}(s^*)
\end{equation}
the instantaneous ground state  would still be the delocalized state and the closing gap between the local and global minimum would not lead to a ground state transition. Hence, the second condition is that the crossing between the global and local minimum has to occur at a time $s^*$ when
\begin{equation}
    E_{\mathrm{local}}(s^*) = E_{\mathrm{global}}(s^*) < E_{\mathrm{deloc}}(s^*) \ .
\end{equation} 
We will refer to the transition from the delocalized state to either of the localized states as delocalized-localized transition, while the transition from one localized state to another is referred to as localized-localized transition.\\
By our assumptions, all the $H_D$ in the class of Hamiltonians Eq. \eqref{eqHdDefinition} that we consider here have the unique ground state
\begin{equation}
    |\psi_0 \rangle = \frac{1}{\sqrt{N}} \sum_z |z \rangle
\end{equation}
with eigenvalue $E_0 = -1$. To obtain $E_{\mathrm{deloc}}(s)$, we reverse the roles of the Hamiltonians and treat $H_D$ as the unperturbed and $H_T$ as the perturbing Hamiltonian. Using first-order non-degenerate perturbation theory we find that
\begin{equation}
\begin{aligned}
    E_{\mathrm{deloc}}(s) &= -(1-s) + s \langle \psi_0 | H_T | \psi_0 \rangle \\
    &= -(1-s) + s \langle E_T \rangle
\end{aligned}
\end{equation}
with
\begin{equation}
    \langle E_T \rangle = \frac{1}{N} \sum_z E_z^T \ .
\end{equation}
At $s=0$, $E_{\mathrm{deloc}}$ is the minimal energy, so the system will be in state $| \psi_0 \rangle$. As $s$ increases, $E_{\mathrm{deloc}}(s)$ will cross either $E_{\mathrm{global}}(s)$ or $E_{\mathrm{local}}(s)$ first. This crossing will demarcate a delocalized-localized transition. If $E_{\mathrm{deloc}}(s)$ crosses $E_{\mathrm{global}}(s)$ first, the ground state will transition directly from the delocalized state to the global minimum (compare Figure \ref{figLevelCrossingPrinciple} (a)). In case $E_{\mathrm{deloc}}(s)$ first crosses $E_{\mathrm{local}}(s)$, the ground state transitions first from the delocalized state to the local minimum. At a later time $s^*$ when $E_{\mathrm{local}}(s)$ crosses $E_{\mathrm{global}}(s)$, there will be an additional localized-localized transition from the local to the global minimum, as depicted in Figure \ref{figLevelCrossingPrinciple} (b).\\
The location of the localized-localized transition $s^*$ is given by the crossing of $E_{\mathrm{global}}(s)$ and $E_{\mathrm{local}}(s)$. Assuming the ground state of $H_T$ to be non-degenerate, $E_{\mathrm{global}}(s)$ can be computed with non-degenerate perturbation theory
\begin{equation}
    E_{\mathrm{global}}(s) = sE_0^T \ .
    \label{eqEglobal}
\end{equation}
The non-degeneracy of the target ground state $E_0^T$ is a simplifying assumption, but our analysis can be easily extended to degenerate target ground states by estimating $E_{\mathrm{global}}(s)$ using the bounds Eq. \eqref{eqEVbound} analogously to the estimate of $E_{\mathrm{local}}(s)$, as we will discuss now.\\
If the degenerate first excited space of $H_T$ is $V$, we can bound $E_{\mathrm{local}}(s)$ using Eq. \eqref{eqEVbound} with $E_V^T = E_1^T$. By solving Eq. \eqref{eqLocLocCondition}, this results in the following bounds on $s^*$
\begin{equation}
    \frac{1-\frac{\phi(V)}{d}}{1-\frac{\phi(V)}{d} + \Delta E^T} \leq s^* \leq \frac{d_{\mathrm{max}}(V)}{d_{\mathrm{max}}(V) + d\Delta E_T}
    \label{eqSstarBounds}
\end{equation}
with the spectral gap of $H_T$
\begin{equation}
    \Delta E_T = E_1^T - E_0^T \ .
\end{equation}
The lower bound depends on the conductance of $V$, while the upper bound depends on the maximum degree. We will refer to these bounds as the conductance and the degree bound respectively. The bounds Eq. \eqref{eqSstarBounds} are the second key result of this work.\\
To predict if there will be a phase transition, we also need to know the location of the delocalized-localized transition to the global minimum $s'$ by solving
\begin{equation}
    E_{\mathrm{deloc}}(s') = E_{\mathrm{global}}(s')
\end{equation}
which renders
\begin{equation}
    s' = \frac{1}{1 + \langle E_T \rangle - E_0^T} \ .
    \label{eqSprime}
\end{equation}
This allows to classify problem instances into three categories:
\begin{enumerate}
    \item The instance has no first-order QPT due to a localized-localized transition (Figure \ref{figLevelCrossingPrinciple} (a)) if
    \begin{equation*}
        s' \geq \frac{d_{\mathrm{max}}(V)}{d_{\mathrm{max}}(V) + d\Delta E_T} \ .
    \end{equation*}
    \item The instance has a first-order QPT due to a localized-localized transition(Figure \ref{figLevelCrossingPrinciple} (b)) if
    \begin{equation*}
        s' \leq \frac{1-\frac{\phi(V)}{d}}{1-\frac{\phi(V)}{d} + \Delta E_T} \ .
    \end{equation*}
    \item If $s'$ is in between the bounds, no statement about first-order QPTs can be made based on the bounds we derived here. We will refer to these problem instances as undecidable.
\end{enumerate}
Note that our first-order analysis allows to estimate the value $s^*$ where a level crossing occurs and provides a qualitative understanding of the conditions leading to first-order QPTs. However, it neglects higher-order interactions of the energy levels that would lift the degeneracy at $s^*$ and lead to an avoided level crossing instead. While these effects are essential for tunneling to occur, they are not required to analyze the conditions when tunneling becomes necessary in the first place. As will be discussed in Section \ref{secNumerics}, our analysis allows to estimate the location of the minimal spectral gap along the annealing path, while the size of the minimal gap can be estimated using higher-order perturbation theory \cite{AminChoi_2009}.
\begin{figure}
    \centering
    \includegraphics[scale=0.42]{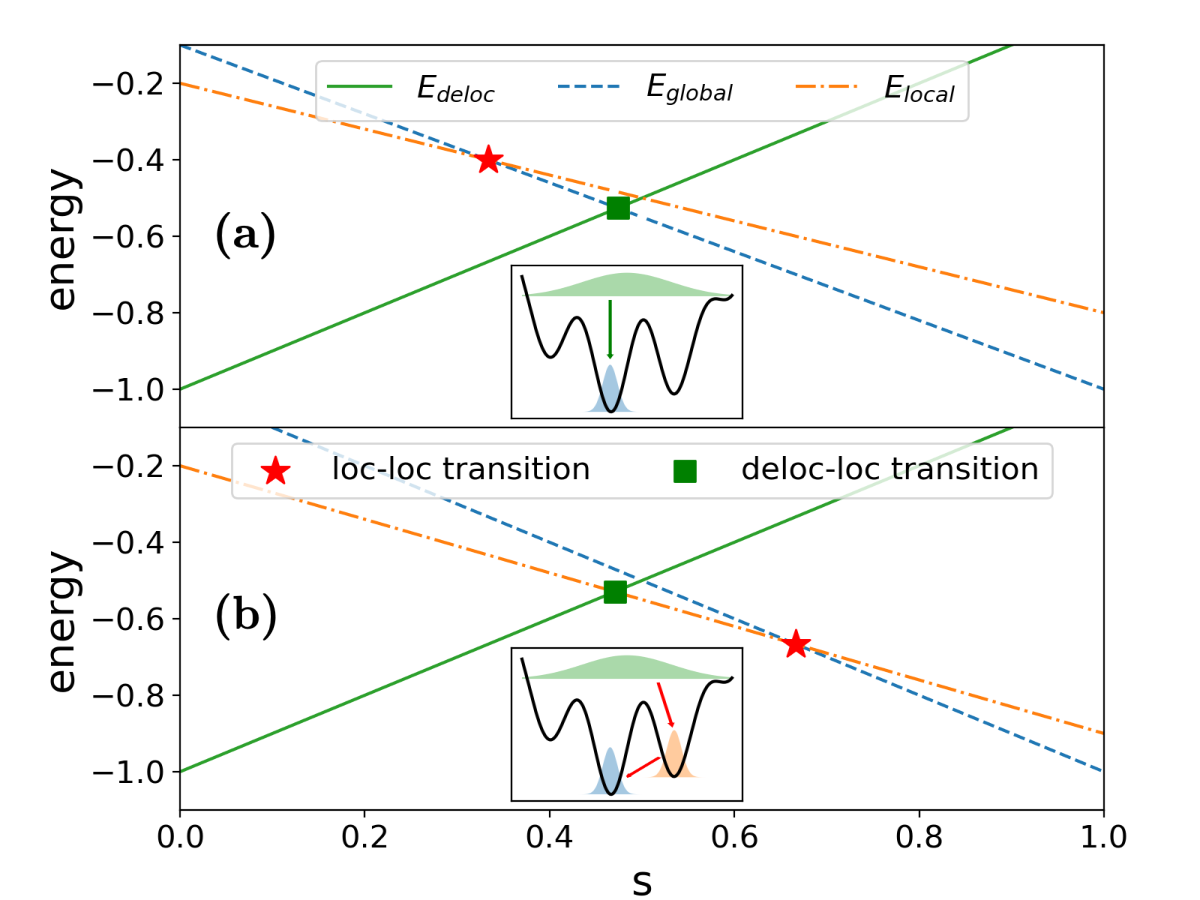}
    \caption{Cartoon of approximate energies with level crossings between the localized states (red stars) and between localized states and the delocalized state (green squares). The localized states correspond to energy $E_{\mathrm{global}}$ (blue dashed line) and $E_{\mathrm{local}}$ (orange dash-dotted line) respectively, while the delocalized state has the energy $E_{\mathrm{deloc}}$ (solid green line). \textbf{(a):} the crossing of the localized states occurs at a time $s$ when the system is still delocalized, hence the ground state will transition to the global minimum directly and only the delocalized-localized transition will be observed. \textbf{(b):} the system transitions from the delocalized state first to the local minimum and subsequently has to tunnel into the global minimum.}
    \label{figLevelCrossingPrinciple}
\end{figure}

\subsection{Correction of the degree bound using graph symmetries}
\label{subsecCorrection}
We will discuss how graph symmetries of $G(V)$ can be used to improve the degree bound on the principal eigenvalue $\lambda_V$. A more detailed discussion of this approach can be found in Appendix \ref{secDetailedSymmetryBound}.\\
If $G(V)$ is an undirected simple graph, then $\lambda_V \leq d_{\mathrm{max}}(V)$, as discussed above. The graph symmetries of $G(V)$ are represented by permutation matrices $\Pi$ \cite{Mowshowitz2009}.
\begin{definition} (Permutation matrix)
    A permutation matrix has in each row and each column one entry 1 and 0 in all other entries. Together with the standard matrix product the permutation matrices form a group $\mathrm{Sym}(V)$.
    \label{defPermutationMatrix}
\end{definition}
Permutation matrices are orthogonal and bijectively map the set of nodes onto itself
\begin{equation}
    \Pi |z\rangle = |z'\rangle \ .
    \label{eqPermutationAction}
\end{equation}
The permutation matrices that commute with the adjacency matrix of a graph form the automorphism group of said graph.
\begin{definition} (Automorphism group)
    Let $G(V) = (V, E)$ a graph. The automorphism group of $G(V)$ is denoted by $\mathcal{S}_V \subseteq \mathrm{Sym}(V)$ and defined as
    \begin{equation}
    \mathcal{S}_V = \{ \Pi \in \mathrm{Sym}(V) : 
    [\Pi, A_{G(V)}] = 0 \} \ .
    \label{eqGraphCommutation}
\end{equation}
\label{defAutomorphismGroup}
\end{definition}
Note that, by definition, the elements of the automorphism group $\mathcal{S}_V$ conserve the neighborhood relations of $G(V)$, as for two nodes $|z_1 \rangle, |z_2\rangle \in V$ and their images $|z'_1 \rangle = \Pi |z_1\rangle$ and $|z'_2\rangle = \Pi |z_2 \rangle$ it holds that
\begin{equation}
    \left( A_{G(V)} \right)_{z_1 z_2} = \left( A_{G(V)} \right)_{z'_1 z'_2} \ .
\end{equation}
Consider
\begin{equation}
    |x\rangle = \sum_z a_z |z\rangle
    \label{eqGeneralEigenvector}
\end{equation}
an eigenvector of $A_{G(V)}$ with eigenvalue $\lambda_V$. Since $A_{G(V)}$ commutes with every permutation matrix $\Pi \in \mathcal{S}_V$, $\Pi |x\rangle$ must also be an eigenvector of $A_{G(V)}$ with eigenvalue $\lambda_V$. However, if the eigenspace of $\lambda_V$ is one-dimensional, $\Pi |x\rangle$ must be proportional to $|x\rangle$, implying that
\begin{equation}
    \Pi |x\rangle = \lambda_\Pi |x\rangle
    \label{eqPermutationEigenvalue}
\end{equation}
for some eigenvalue $\lambda_{\Pi}$ with $|\lambda_\Pi| = 1$. That $|\lambda_\Pi| = 1$ follows, since all $\Pi \in \mathrm{Sym}(V)$ are orthogonal matrices. Hence, any non-degenerate eigenvector of $A_{G(V)}$ must also be an eigenvector of all permutation matrices that commute with $A_{G(V)}$. In order to respect Eq. \eqref{eqPermutationEigenvalue}, the coefficients $a_z$ in Eq. \eqref{eqGeneralEigenvector} of nodes connected by some symmetry of the graph need to have the same amplitude and a fixed phase relation. Assuming $G(V)$ is connected, the principal eigenvalue $\lambda_V$ is non-degenerate and the principal eigenvector can be chosen with all-positive coefficients, according to the Perron-Frobenius theorem. This means that the principal eigenvalue lies in the subspace spanned by
\begin{equation}
    |\xi \rangle = \frac{1}{\sqrt{|\xi|}} \sum_{z \in \xi} |z\rangle
\end{equation}
where the $\xi$ are the sets of nodes that are mapped onto each other by some symmetry. As discussed in Appendix \ref{secDetailedSymmetryBound}, the $\xi$ are equivalence classes of nodes.\\
The adjacency matrix elements in this subspace have the form
\begin{equation}
    \langle \xi| A_{G(V)} | \xi' \rangle = \frac{1}{\sqrt{|\xi||\xi'|}} \sum_{z \in \xi, z' \in \xi' } \left( A_{G(V)} \right)_{z,z'} \ .
\end{equation}
Applying Gershgorin's circle theorem to this matrix, we find
\begin{equation}
    \lambda_V \leq \max_{\xi} \sum_{\xi'} |\langle \xi| A_{G(V)} | \xi' \rangle |
\end{equation}
which evaluates to
\begin{equation}
    \lambda_V \leq \max_{\xi} \sum_{\xi'} \sqrt{ \frac{|\xi|}{|\xi'|}} |E_{\xi \xi'}|
    \label{eqGersh1}
\end{equation}
where $|E_{\xi \xi'}|$ is the number of nodes of equivalence class $\xi'$ in the neighborhood of a node of equivalence class $\xi$.\\
For comparison, in the computational basis $|z\rangle$, Gershgorin renders the bound
\begin{equation}
    \lambda_V \leq \max_{\xi} \sum_{\xi'} |E_{\xi \xi'}| \ .
    \label{eqGersh2}
\end{equation}
The right-hand side of Eq. \eqref{eqGersh1} can be smaller than the right-hand side of Eq. \eqref{eqGersh2}, thus symmetries of the graph can tighten the bound.\\
It is, in fact, reasonable to expect symmetries to move the bounds Eq. \eqref{eqPrincipalEigenvalueBounds} closer to each other, as the conductance bound is based on the uniform superposition of all nodes in $V$ as a variational ansatz (see Appendix \ref{secAppendixDerivationOfBounds}). The uniform superposition naturally is invariant under all permutations of nodes in $V$.\\

\section{Numerical investigation}
\label{secNumerics}
\subsection{Simple toy model}
We will first analyze the developed bounds in an idealized toy model. To this end, we generate random $d$-regular simple graphs of size $|\mathcal{V}| = 256$ and $d=8$. We place the local and global maxima as far apart on the graph as possible, where the distance is measured in terms of traversed edges in the graph, and iteratively grow the local minimum by randomly selecting a node $i \in \mathcal{N}(V)$ and add it to the set $V \leftarrow \{ i \} \cup V$. Here, $\mathcal{N}(V)$ denotes the neighborhood of $V$, i.e. any node in $\mathcal{V} \setminus V$ that shares at least one edge with any node in $V$.\\
As the energy of the global minimum we choose $E_0^T=-1$, while the energy of the local minimum is chosen as $E_V^T = -1 + \Delta E^T$ with $\Delta E^T$ sampled uniformly between 0 and 1. For all but the initial node in the set $V$ the energies are furthermore slightly increased by $\epsilon = 0.01$. This is required to make non-degenerate perturbation theory applicable for comparison, as discussed below. Following this procedure, we generate energy landscapes on random regular graphs with a narrow, non-degenerate global minimum and a wide, arbitrarily shaped and nearly degenerate local minimum far away from the global minimum.\\
\begin{figure}
    \centering
    \includegraphics[scale=0.5]{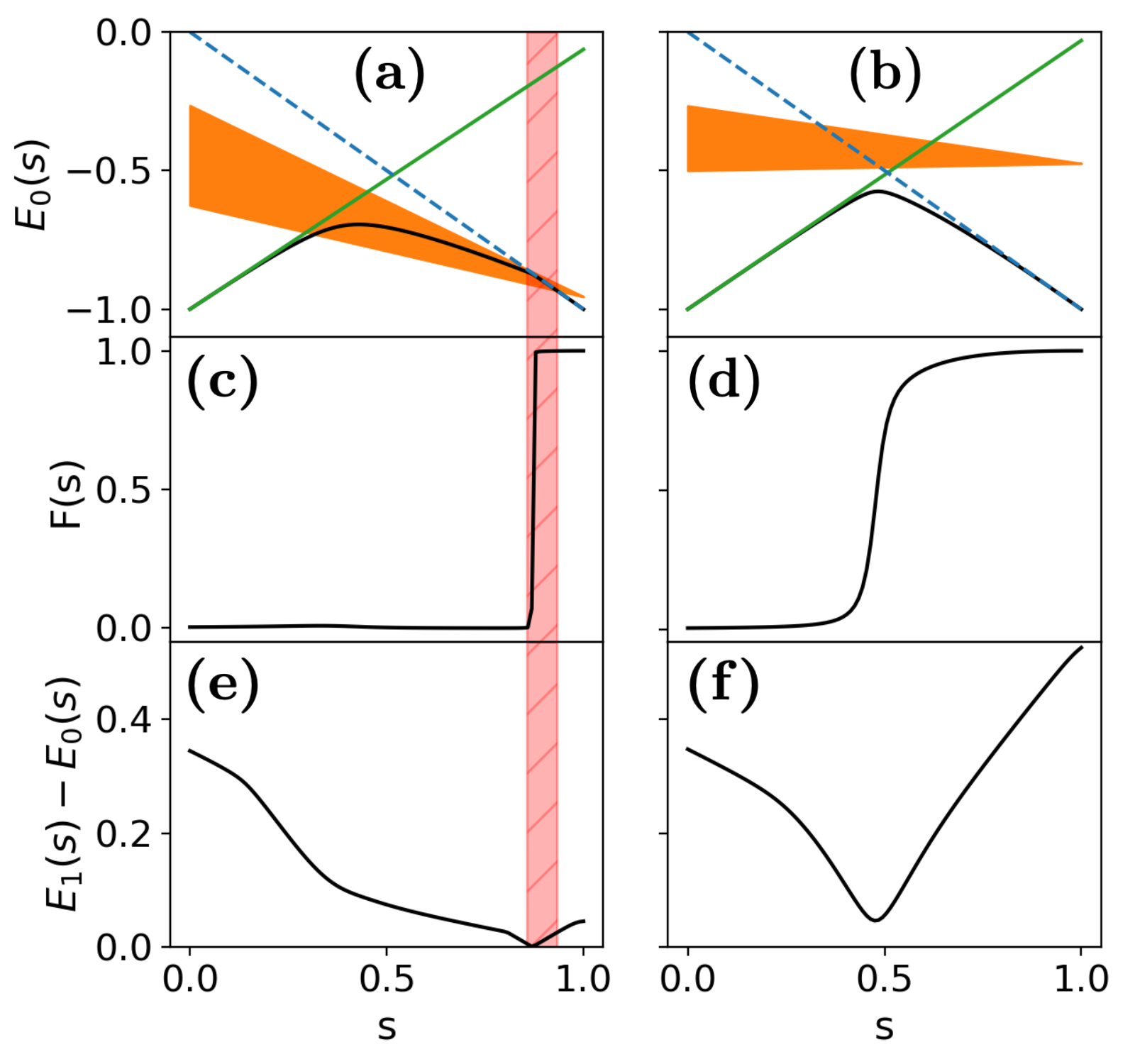}
    \caption{\textbf{(a):} Instantaneous ground state energy $E_0(s)$, \textbf{(c):} solution fidelity $F(s)$ and \textbf{(e):} instantaneous spectral gap $E_1(s) - E_0(s)$ with a first-order QPT. The true ground state energy (black) follows closely the predicted energies of the delocalized state (solid green line), the local minimum (orange shaded area) and the global minimum (dashed blue line). The energy of the local minimum is given as a shaded area according to the bounds Eq. \eqref{eqEVbound}. The solution fidelity is discontinuous within the predicted bounds Eq. \eqref{eqSstarBounds} (red shaded striped interval) coinciding with the minimal spectral gap. \textbf{(b)}, \textbf{(d)} and \textbf{(f)} show the respective quantities for a problem instance without a first-order QPT. Again, the true ground state energy follows the predicted energies, the solution fidelity is smoother, while the spectral gap is larger. All figures with $d=8$ and $|\mathcal{V}|=256$.}
    \label{figNumericalExample}
\end{figure}%
The first-order QPTs can be identified easily by looking at the solution fidelity $F(s)$ along the anneal, which is defined as the overlap between the instantaneous ground state $|\Psi_0(s)\rangle$ and the target ground state $|0\rangle$
\begin{equation}
    F(s) = |\langle \Psi_0(s) | 0 \rangle|^2 \ .
\end{equation}
In Figure \ref{figNumericalExample} we show the ground state energy $E_0(s)$, the solution fidelity $F(s)$ and the spectral gap between instantaneous ground and first excited state $E_1(s) - E_0(s)$ over the annealing of a toy model instance with (a, c, e) and without (b, d, f) first-order QPT according to the classification based on the conductance and degree bounds Eq. \eqref{eqSstarBounds}. Note that we are analyzing the spectral properties of the instantaneous Hamiltonian $H(s)$ as a function of $s$, rather than concrete dynamics of a system.\\
We observe that the true ground state energy follows closely the respective minimum of the perturbed energies of the delocalized state and local and global minima. As the energy of the local minimum is bounded by Eq. \eqref{eqEVbound}, it is depicted as a shaded area (orange). The bounds on the transition point between the local and global minimum from Eq. \eqref{eqSstarBounds} are shown as the shaded, striped interval (red) and bound the location of the abrupt jump in $F(s)$ as well as the location of the minimal spectral gap as seen in Figure \ref{figNumericalExample} (a, c, e). In the case where the bounds predict an absence of a first-order QPT (Figure \ref{figNumericalExample} b, d, f), the true ground state energy is well described by the delocalized energy and the perturbed global minimum. The solution fidelity is smooth and the minimal spectral gap is significantly larger.\\
Following the procedure described above we generate several problem instances and predict the bounds on the localized-localized transition. As first-order QPTs are typically associated with an exponentially closing gap, it is reasonable to assume the location of the minimal spectral gap $s_{\mathrm{min}}$ to coincide with the location of the QPT $s^*$. Therefore, we can test the theory by comparing the predictions of $s^*$ with $s_{\mathrm{min}}$ obtained from exact diagonalization of the instantaneous Hamiltonian $H(s)$.\\
In Figure \ref{figToyModelResults}, we show the predicted bounds over the observed minimal spectral gap from exact diagonalization. The diagonal (red dashed line) denotes the equality of predicted and observed $s_{\mathrm{min}}$. If a point is close to the diagonal, it means that the predicted and the observed $s_{\mathrm{min}}$ are close to each other. The derived bounds Eq. \eqref{eqSstarBounds} reliably cross the diagonal, indicating that the true $s_{\mathrm{min}}$ is within the predicted bounds. We further classify the problem instances according the presence (blue) or absence (red) of first-order QPTs, as well as the undecidable instances (orange), as discussed in section \ref{secTheorySpecTheo}. The problem instances with QPT are well described by the bounds, as well as the ambiguous cases. For the instances without QPT, $s_{\mathrm{min}}$ is not
predicted well. This can be explained as we associate the minimum with the QPT, but the instances do not display a level-crossing between the local and global minimum, as the localized-localized transition would happen at a value of $s < s'$ when the ground state is still delocalized. Interestingly, in particular for the problem instances with a localized-localized transition, the lower bound on $s_{\mathrm{min}}$, i.e. the conductance bound, seems to be much closer to the true observed $s_{\mathrm{min}}$ than the degree bound.\\
The phenomenon of first-order QPTs has been investigated in a previous work by Amin et al. \cite{AminChoi_2009}, where they employ second-order non-degenerate perturbation theory to calculate the location of the level-crossing between local and global minimum. Non-degenerate perturbation theory diverges for degenerate eigenstates, making the predictions less accurate as the local minima become wide. As first-order QPTs are associated with wide local minima \cite{King_2016}, we would expect degenerate perturbation theory to better describe the relevant energy corrections. Furthermore, as we argue here, degenerate perturbation theory leads to an interesting unification with spectral graph theory.\\
For comparison we apply their predictions to the same problem instances as well. As described before, the first state in the set $V$ has its energy set to $1+\Delta E_T$ and the energies of all other states in $V$ are set to $1+\Delta E_T + \epsilon$. Here we use $\epsilon=0.01$. Then we can take the first state in $V$ as the local minimum and compute its energy correction by coupling to its neighbors using second-order non-degenerate perturbation theory. Here it becomes apparent why the additional $\epsilon$ is necessary, since for degenerate neighboring states, the second-order non-degenerate energy corrections would diverge and not render meaningful predictions.\\
The predictions from non-degenerate perturbation theory are also shown in Figure \ref{figToyModelResults} (green crosses). Figure \ref{figToyModelResults} shows that the predictions based on non-degenerate perturbation theory are further away from the diagonal, implying that the predictions are less accurate. Qualitatively, first-order non-degenerate perturbation theory does not predict level crossings for $s \in (0, 1]$, as discussed in \cite{AminChoi_2009}. Our approach shows that the energy corrections leading to first-order QPTs can be described by first-order degenerate perturbation theory, while non-degenerate perturbation theory requires second order corrections, i.e. the relevant energy corrections are a first-order, rather than second order effect.

\begin{figure}
    \centering
    \includegraphics[scale=0.5]{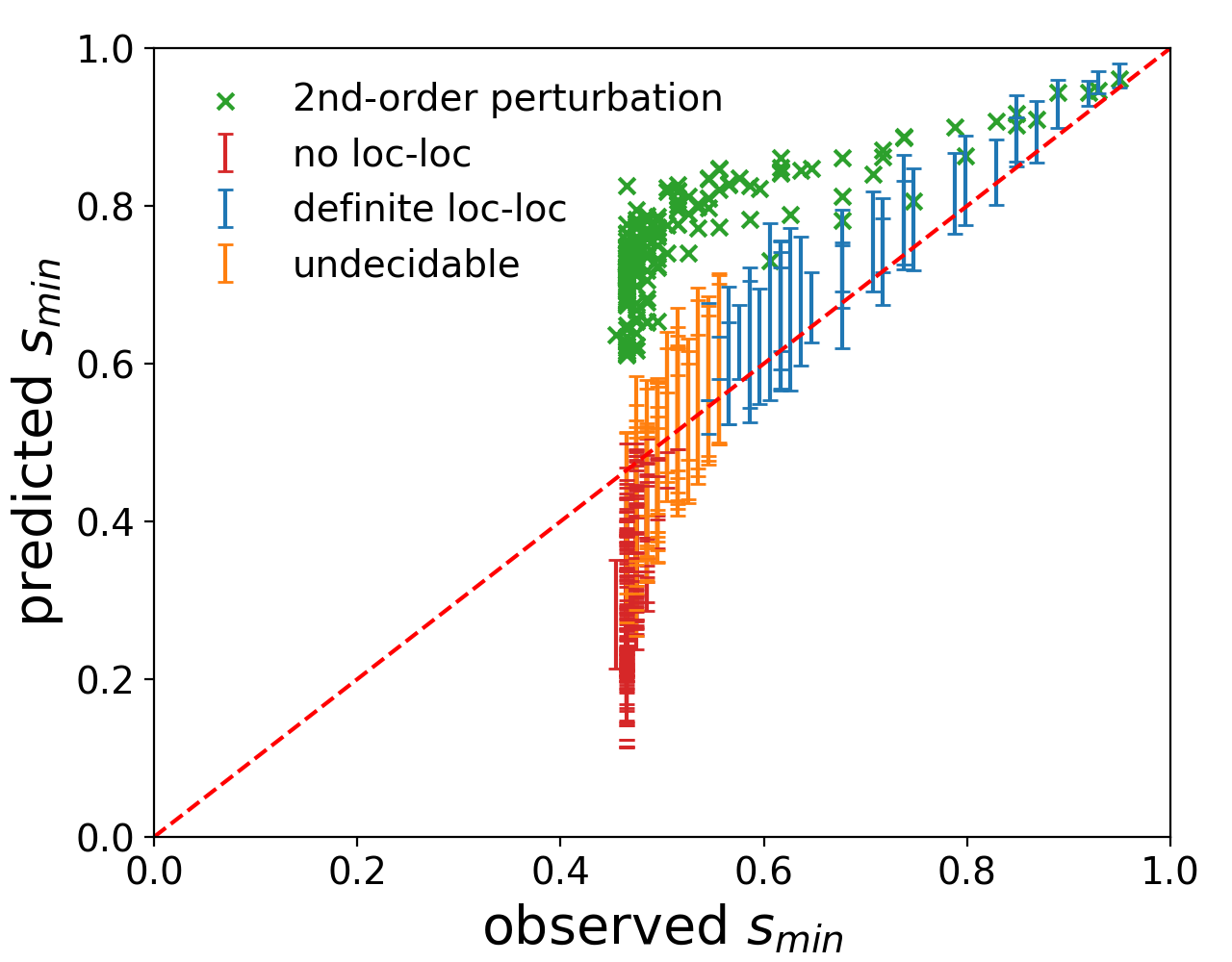}
    \caption{Predicted location of the minimal spectral gap $s_{\mathrm{min}}$ according to Eq. \eqref{eqSstarBounds} over the observed location from exact diagonalization for $d=8$ and $|\mathcal{V}| = 256$. The bounds for instances with a first-order QPT (blue) reliably cross the diagonal (red dashed line), indicating that the true $s_{\mathrm{min}}$ is within the predicted bounds Eq. \eqref{eqSstarBounds}. For instances without a QPT (red) the true value is outside of the bounds, as these instances do not exhibit a localized-localized transition. For comparison the predictions of the location of the QPT according to second-order non-degenerate perturbation theory according to \cite{AminChoi_2009} are shown (green crosses). The predictions from non-degenerate perturbation theory are further away from the observed values, especially if the observed $s_{\mathrm{min}}$ is further away from 1.}
    \label{figToyModelResults}
\end{figure}

\subsection{Weighted Minimum Independent Set}
Lastly, we apply our analysis to a problem instance of an NP-complete problem, namely the Weighted Maximum Independent Set (WMIS) problem. As our analysis requires extensive knowledge of the energy landscape and eigenstates of $H_T$ in order to compute the relevant quantities, it is infeasible to apply it in a more automatized manner. For this reason, as well as for direct comparison with the results of Amin et al., we use the same problem instance as in \cite{AminChoi_2009}.\\
Consider a vertex-weighted graph $G_P = (\mathcal{V}_P, \mathcal{E}_P, w)$, where $w: i \rightarrow w(i)$ is the weight of node $i \in \mathcal{V}_P$. This graph defines a problem instance and the nodes are identified with qubits. This is a strictly distinct type of graph from the graphs used in the theoretical analysis, where the nodes are individual basis states. This distinction is highlighted by the subscript $P$.\\
The WMIS problem is to find the largest subset $S \subseteq \mathcal{V}_P$ such that no two nodes in $S$ share an edge (i.e. it is independent), while simultaneously maximizing the weight
\begin{equation}
    w(S) = \sum_{i \in S} w(i) \ .
\end{equation}
This optimization problem can be cast into the target Ising Hamiltonian
\begin{equation}
    H_T = \sum_{i \in \mathcal{V}_P} h_i \sigma_i^z + \sum_{i, j \in \mathcal{E}_P} J_{ij} \sigma_i^z \sigma_j^z
\end{equation}
with the fields and couplings being
\begin{equation}
\begin{aligned}
    h_i &= \sum_{i,j \in \mathcal{E}_P} J_{ij} - 2w(i) \ , \\
    J_{ij} &> \min\{ w(i), w(j) \} \ .
\end{aligned}
\end{equation}
Amin et al. use two different weights on the nodes, $w(i) = w_G = 1$ for all nodes that partake in the optimal solution and $w(i) = w_L <2w_G$ for all nodes outside the optimal solution. The couplings are chosen as $J_{ij} = 2$. The parameter $w_L$ can be altered to change the depth of the local minima. The problem graph has $N_q = 15$ nodes, each represented by a single qubit. The driver Hamiltonian is the transverse-field driver from Eq. \eqref{eqSigmaXDriver} and its associated graph is the $N_q$-dimensional hypercube. $H_T$ has 27 shallow local minima. These local minima are separated by two bit-flips, allowing for tunneling between the local minima. Therefore, we can consider the 27 local minima plus the shallow potential walls between them as one nearly degenerate local minimum $V$. Given this knowledge of the locations of the local minima, the relevant quantities for $V$ can be counted
\begin{equation}
\begin{aligned}
    |V| &=  135 \ , \\
    |\partial V| &= 1539 \ , \\
    d_{\mathrm{max}}(V) &= 9 \ , \\
    \Delta E_T &= 4(6w_G - 3w_L) \ .
\end{aligned}
\end{equation}
Note that $H_D$ and $H_T$ for the WMIS do not adhere to our assumptions about the normalization of the ground state energies. We adapt the expressions Eq. \eqref{eqSstarBounds} and Eq. \eqref{eqSprime} accordingly by dropping the respective normalization factors in the derivation. For more details on the problem graph $G_P$ we refer to the original publication \cite{AminChoi_2009}.\\
\begin{figure}
    \centering
    \includegraphics[scale=0.42]{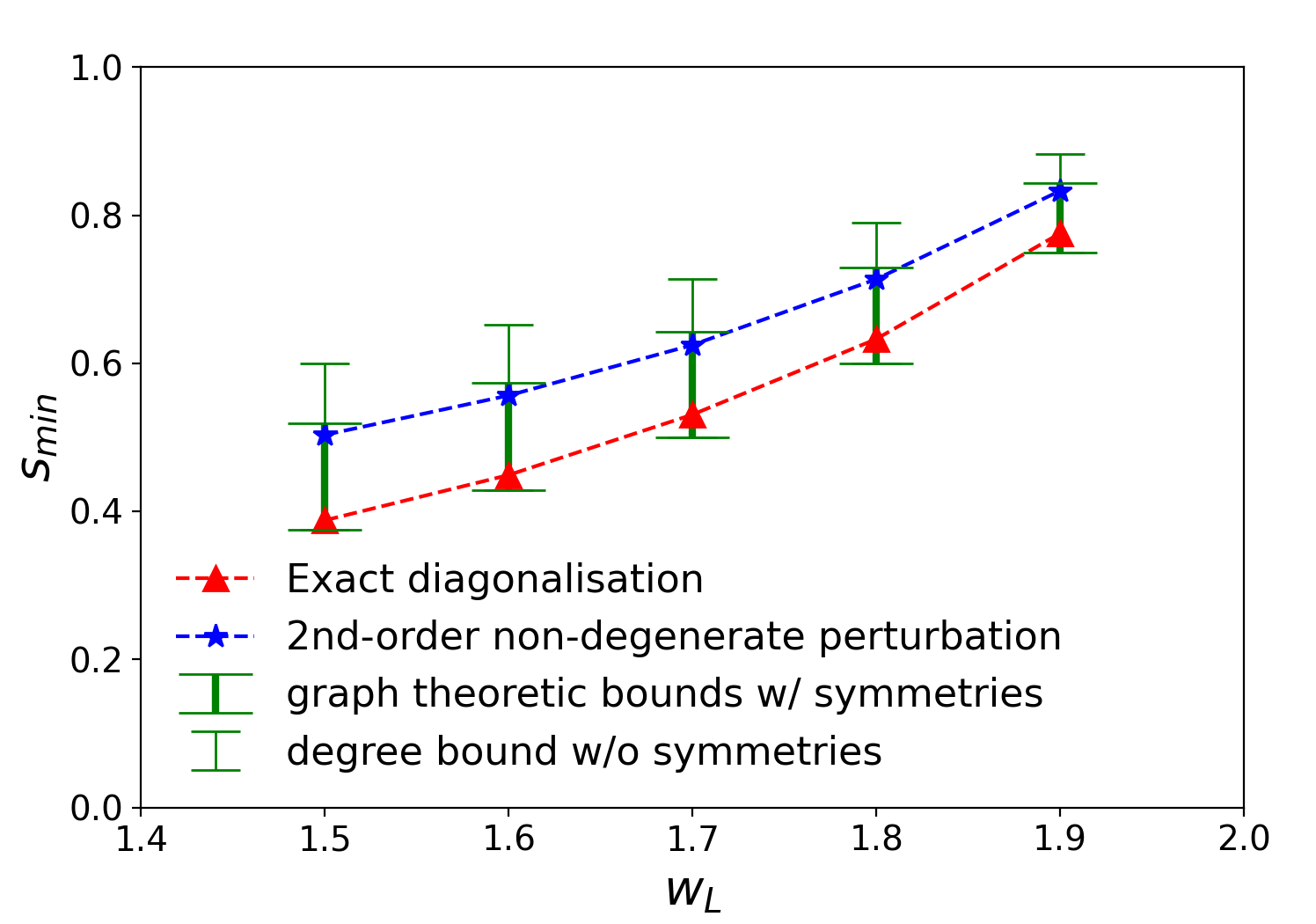}
    \caption{Location of the minimal spectral gap along the annealing path $s_{\mathrm{min}}$ for a WMIS instance \cite{AminChoi_2009} obtained via exact numerical diagonalization (red triangles), by second-order non-degenerate perturbation theory (blue stars) and by the graph theoretic method introduced here (green error bars). The upper bound, i.e. the degree bound, on $s_{\mathrm{min}}$ is shown once without considering graph symmetries of $G(V)$ (narrow caps) and once after applying the corrections using graph symmetries (wide caps). Remarkably, the conductance bound matches very well the exact result. The upper bound is fixed using the bound improved by graph symmetries.}
    \label{figWMISResultsSymmetry}
\end{figure}%
The adapted expressions allow to make predictions of the location of the minimal spectral gap along the annealing path $s_{\mathrm{min}}$. In Figure \ref{figWMISResultsSymmetry} we show the graph theoretic bounds as well as the predictions by second-order non-degenerate perturbation theory by Amin et al. for various depths of the local minimum $w_L$, together with the exact numerical diagonalization. From this comparison, we observe that the exact results are within the bounds that we define and are very close to the lower one, i.e. the conductance bound, as was the case in the toy model in the previous discussion.\\
The degree bound can be tightened by taking the graph symmetries of $G(V)$ into account. We find the local minimum $V$ of the WMIS instance has nodes of three equivalence classes, as depicted by the black circles, red squares and blue triangles in Figure \ref{figLocalMinGraph}.
\begin{figure}
    \centering
    \includegraphics[scale=0.25]{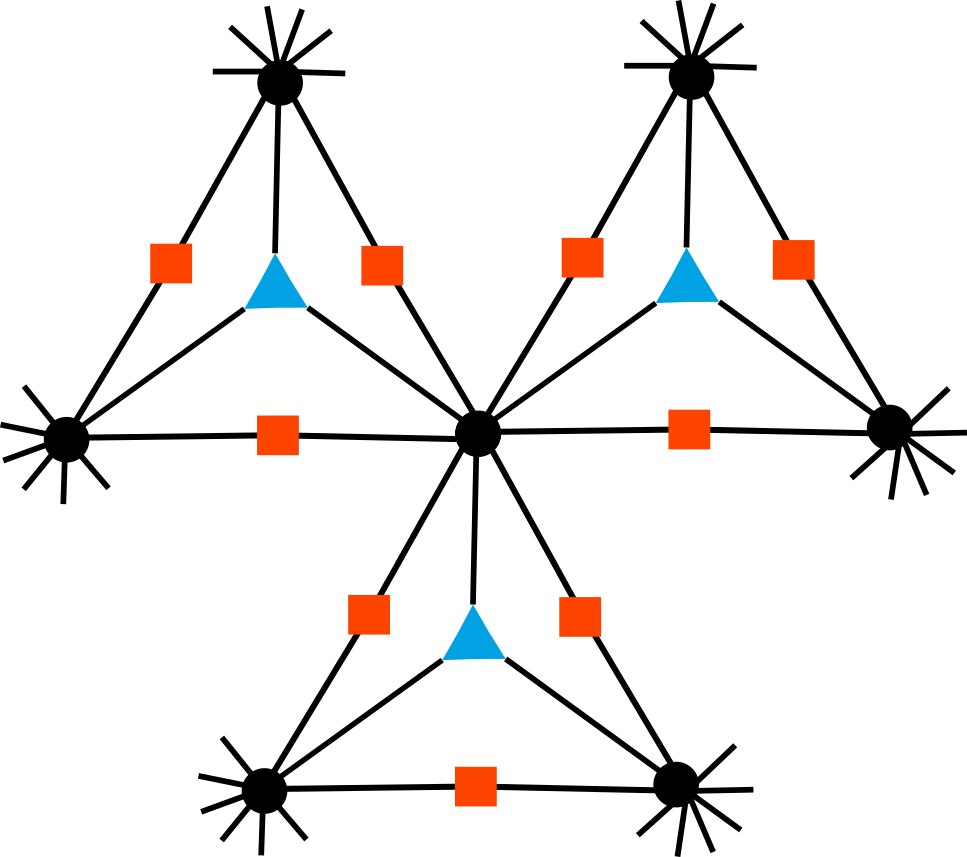}
    \caption{Induced subgraph $G(V)$ of the local minimum of the WMIS instance from \cite{AminChoi_2009}, black dots corresponding to states with one qubit in the $|0\rangle$-state in each of the outer cliques defining the local minimum, red squares representing states with one qubit in the $|0\rangle$-state in two cliques and two qubits in the $|0\rangle$-state in the third, and blue triangles corresponding to states with one qubit in the $|0\rangle$-state in two cliques and no qubit in the $|0\rangle$-state in the third.}
    \label{figLocalMinGraph}
\end{figure}
The black circles are the nodes with the largest degree, hence the previous upper bound $\lambda_V \leq d_{\mathrm{max}} = 9$. There are 27 black circles, 27 blue triangles and 81 red squares. As there are no connections within equivalence classes in $G(V)$, $\langle \xi | A_{G(V)} | \xi \rangle = 0$. Hence, in the basis of the equivalence classes the adjacency matrix reads
\begin{equation}
\begin{aligned}
    &\sum_{\xi, \xi'} \langle \xi | A_{G(V)} | \xi' \rangle |\xi \rangle \langle \xi' | \\ 
    &= \begin{pmatrix} 
    0 & 2\sqrt{3} & 3 \\
    2\sqrt{3} & 0 & 0 \\
    3 & 0 & 0
    \end{pmatrix} \ .
\end{aligned}
\end{equation}
Applying Gershgorin's circle theorem results in
\begin{equation}
    \lambda_V \leq 2\sqrt{3} + 3 = 6.46... \ .
\end{equation}
Using this improved estimate of the principal eigenvalue of $G(V)$, we get improved estimates of the location of the minimal spectral gap, as shown in Figure \ref{figWMISResultsSymmetry}.

\section{Discussion}
\label{secDiscussion}
\subsection{Tightness of the bounds and Cheeger inequalities}
It is possible to obtain an intuition about the tightness of the bounds on $E_V(s)$ and thus about the bounds on $s^*$.\\
As we find for the average degree $\langle d(V) \rangle$ of nodes in $G(V)$ (see Appendix \ref{secAppendixDerivationOfBounds})
\begin{equation}
    \langle d(V) \rangle = d - \phi(V)
\end{equation}
we can conclude for regular induced subgraphs $G(V)$ that the bounds Eq. \eqref{eqPrincipalEigenvalueBounds} turn into equalities
\begin{equation}
    d-\phi(V) = d_{\mathrm{max}}(V) = \lambda_V \ .
\end{equation}
As a consequence, the bounds Eq. \eqref{eqEVbound} on $E_V(s)$ are equal and therefore exact within first-order perturbation.\\
However, $G(V)$ does not exactly need to be regular for the conductance bound Eq. \eqref{eqEVupperbound} to become tight. It can be shown that the principal eigenvalue of large Erdös-Rényi graphs $G(n,p)$ is approximately $np$ \cite{Krivelevich_2001}, i.e. the average degree $\langle d \rangle$. Therefore, if $G(V)$ can be considered a large, sparse random graph, the conductance bound is probably tight.\\
Employing further results from spectral graph theory, the conductance $\phi$ admits an interesting connection to the spectral gap of $H_D$. Let us define a non-trivial lower bound to the conductance of an undirected graph $G= (\mathcal{V}, \mathcal{E} )$
\begin{equation}
    \phi_0 = \min_{\substack{U \subset \mathcal{V} \\ U \neq \emptyset \\ |U| \leq N/2}} \phi(U) \ .
\end{equation}
In words, $\phi_0$ is the smallest conductance over all non-empty subsets of $\mathcal{V}$ containing at most half of all nodes. This quantity is called the Cheeger constant \cite{Chung_1997}, also known as the conductance of the graph $G$. The Cheeger constant can be linked to the spectral gap of a $d$-regular graph's adjacency matrix $A_G$ via the Cheeger inequalities \cite{Chung_1997}. Since the $H_D$ considered here are proportional to a $d$-regular graph's $A_G$, the inequalities are easily adapted to give
\begin{equation}
    \frac{1}{2}\frac{\phi_0^2}{d^2} \leq \Delta E_D \leq 2\frac{\phi_0}{d}
\label{eqCheegerInequalities}
\end{equation}
with $\Delta E_D$ the spectral gap of $H_D$.\\
Assuming that the low-energy sub-spaces of $H_T$ under perturbation with $H_D$ for some problem class fulfill the conditions discussed above and the corrected energy eigenvalues are close to the conductance bound, we can determine the location of the localized-localized transition to be close to the conductance bound
\begin{equation}
    s^* \approx \frac{1 - \frac{\phi(V)}{d}}{1 - \frac{\phi(V)}{d} + \Delta E_T} \ .
    \label{eqAssumption}
\end{equation}
Under these conditions, $s^*$ increases monotonously as $\phi(V)$ decreases, allowing us to claim
\begin{equation}
    s^* \leq \frac{1 - \frac{\phi_0}{d}}{1 - \frac{\phi_0}{d} + \Delta E_T} \ .
\end{equation}
This gives a condition for the absence of first-order QPT under the stated assumptions by using Eq. \eqref{eqSprime} and setting $s^* < s'$ to get
\begin{equation}
    \frac{E_1^T - \langle E_T \rangle}{E_0^T - \langle E_T \rangle} < \frac{\phi_0}{d} \ .
    \label{eqNoQPTCondition1}
\end{equation}
Using the Cheeger inequalities Eq. \eqref{eqCheegerInequalities} the condition Eq. \eqref{eqNoQPTCondition1} can be stated in terms of the spectral gap of $H_D$
\begin{equation}
    \frac{E_1^T - \langle E_T \rangle}{E_0^T - \langle E_T \rangle} < \frac{\Delta E_D}{2} \ .
    \label{eqNoQPTCondition2}
\end{equation}
Given assumption Eq. \eqref{eqAssumption} this will render a sufficient condition for the absence of first-order QPTs for any distribution of the energies $E_z^T$ over the nodes of $G$, i.e. the target eigenstates $|z\rangle$, that is consistent with the above assumption.\\
The condition Eq. \eqref{eqNoQPTCondition2} implies that the energy landscape of the problem encoded in $H_T$ needs to be sufficiently flat, except for a pronounced ground state energy. As an example, consider an $H_T$ with $E_0^T = E_1^T - \Delta E_T = -1$ and $\langle E_T \rangle = 0$. Then Eq. \eqref{eqNoQPTCondition2} can be rearranged to read
\begin{equation}
    1 - \frac{\Delta E_D}{2} < \Delta E_T \ .
\end{equation}
For a fixed spectral gap $\Delta E_D$ this implies a lower bound on the spectral gap of $H_T$ and since by assumption we restricted $E_0^T = -1$ and $\langle E_T \rangle = 0$, the spectrum of $H_T$ has to concentrate close to $0$.

\subsection{Interpretation of the bounds}

Let us now present a more physical interpretation of the derived bounds. $d_{\mathrm{max}}(V)$ provides a notion of the maximum number of degrees of freedom involved in the local minimum, e.g. in the case of a transverse field $H_D$ as in Eq. \eqref{eqSigmaXDriver}, $d_{\mathrm{max}}(V)$ describes the number of floppy qubits involved in the minimum. Floppy qubits are qubits that do not significantly impact the energy of the total system, whether they are in the $|0\rangle$- or $|1\rangle$-state. Therefore they cause wide local minima and are known to contribute to perturbative anti-crossings \cite{Dickson_2011a, King_2016}. However, the notion of width of a local minimum in higher dimensions becomes an increasingly poorly defined concept. Our lower bound using the conductance provides a notion of width, as the conductance can be thought of as some version of surface-to-volume ratio of the local minimum, where the volume corresponds to the number of nodes in the local minimum, while the surface corresponds to the edges leaving the local minimum.\\
In general, calculating $\phi(V)$ and $d_{\mathrm{max}}(V)$ of a given problem instance is not going to be scalable, as it requires knowledge of the set of nodes $V$ that constitute a local minimum. Therefore, it is unlikely that the bounds will find direct application in an algorithm to overcome the limitations of AQC. However, we believe that our analysis provides a valuable framework to understand the conditions that lead to the occurrence of QPTs, which in turn can aid the development of tools in their mitigation. Furthermore, there are potential applications of our results in the complexity analysis for AQC.

\section{Conclusion}
First-order QPTs are known to lead to exponential runtime in AQC algorithms, hampering the chances of getting any quantum advantage in the adiabatic computation. Understanding the causation of such phenomena is therefore key in the design of potential strategies to allow for its mitigation. In this work we examine the conditions linked to the occurrence of first-order QPTs caused by localization in AQC. We explicitly show how the use of degenerate perturbation theory enables the application of tools from graph theory in order to analyze this phenomenon in depth. Crucially, this formalism allows us to derive bounds and conditions for the occurrence of QPTs as well as its position along the annealing path. 
We show how such inequalities are linked to two properties of the subgraph containing the local minimum: its maximum degree $d_{\mathrm{max}}(V)$ and its conductance $\phi(V)$.
We numerically test the accuracy of these bounds with a toy-model problem as well as a real optimization problem (WMIS) and find that the lower bound seems to be closer to the exact solution obtained through direct diagonalization of the Hamiltonian. Based on this observation, we provide two scenarios when we can expect the conductance bound to be exact up to first order perturbation theory: namely when the subgraph induced by the degenerate subspace $V$, $G(V)$, is regular or can be assumed to be a large, sparse Erdös-Rényi graph.\\
Additionally we show how knowledge of the symmetries of the induced subgraph $G(V)$ can be used to improve the upper bound on the principal eigenvalue, and further work is required to understand whether this contributes to a tighter conductance bound. After establishing the basis for this formalism, future work will be focused on its application to the study of catalysts in the Hamiltonian as a strategy to enlarge the minimum energy gap and thus improve the performance of AQC algorithms \cite{Albash_2019, Crosson_2020, Choi_2021, Feinstein_2022}. Furthermore, our results may open the door to the construction of gap oracles, which can be used to obtain quadratic speed-up in adiabatic quantum optimization \cite{Jarret2018_2}. We believe that the use of our graph-based approach is not only convenient in the study of QPTs in the context of AQC, but also of a wide range of many-body phenomena present in analog-based quantum computation. 




\section*{ACKNOWLEDGMENTS}
We thank A. Palacios de Luis and J. Riu for valuable discussions. This work was supported by European Commission FET-Open project AVaQus (GA 899561) and the Agencia de Gesti\'o d’Ajuts Universitaris i de Recerca through the DI grant (No. DI75).
\bibliographystyle{unsrt}
\bibliography{references}

\appendix
\section{Derivation of graph bounds}
\label{secAppendixDerivationOfBounds}
For simple undirected graphs the upper bound follows directly from Gershgorin circle theorem, which we will state here for completeness.
\begin{theorem} (Gershgorin Circle Theorem \cite{Gershgorin31})
Let $M$ a complex square matrix with elements $m_{ij}$. All eigenvalues $\lambda$ of $M$ lie in the union of the disks $D(m_{ii}, R_i)$ centered at $m_{ii}$ and with radii
\begin{equation*}
    R_i = \sum_{j\neq i} |m_{ij}| \ .
\end{equation*}
\end{theorem}
From this it follows directly that
\begin{equation}
    |\lambda| \leq \max_i \sum_j |m_{ij}| \ .
    \label{eqGersh}
\end{equation}
Applying Eq. \eqref{eqGersh} to the adjacency matrix $A_G$ of a simple graph $G$ we find for all eigenvalues $\lambda$ of $A_G$
\begin{equation}
    \lambda \leq d_{\mathrm{max}}
\end{equation}
since the $A_G$ of simple graphs have all zeros everywhere except for $(A_G)_{ij} = 1$ in the off-diagonal entries corresponding to connected nodes.\\
The lower bound is usually stated in terms of the average degree and can be shown using Rayleigh quotient for a Hermitian matrix $M$ with maximal eigenvalue $\lambda$
\begin{equation}
    R(M, x) := \frac{x^T M x}{x^Tx} \leq \lambda \ .
\end{equation}
Choosing $M = A_G$ and $x = [1, 1, ...]^T$ shows the lower bound and hence we find
\begin{equation}
    \langle d \rangle \leq \lambda \leq d_{\mathrm{max}}
\end{equation}
where $\langle d \rangle$ is the average degree of all nodes $d_{\mathrm{max}}$ the maximum degree in $\mathcal{V}$. These bounds are a standard result from spectral graph theory.\\
As discussed in the main text, we have to solve the eigenvalue problem on the (nearly) degenerate subspace $V$, which is to say we have to find the principal eigenvalue of the induced subgraph $G(V)$. The maximum degree of all nodes in $G(V)$ is $d_{\mathrm{max}}(V)$, hence the upper bound follows directly. A little more care needs to be taken to link the average degree to the conductance.\\
\begin{lemma}
Let $G = (\mathcal{V}, \mathcal{E})$ a $d$-regular simple graph. Let $G(V) = (V, E) \subseteq G$ the subgraph induced by $V \subseteq \mathcal{V}$ with adjacency matrix $A_{G(V)} \in \{0, 1 \}^{|V| \times |V|}$ and principal eigenvalue $\lambda_V$. Let $\phi(V)$ the conductance of $V$. Then 
\begin{equation*}
    d - \phi(V) \leq \lambda_V \ .
\end{equation*}
\label{lemmaLowerBoundPrincipalEigenvalue}
\end{lemma}
\begin{proof}
Consider the Rayleigh quotient of $A_{G(V)}$
\begin{equation}
    R(A_{G(V)}, x) = \frac{x^T A_{G(V)}x}{x^T x} \leq \lambda_V
    \label{eqSubgraphRayleigh1}
\end{equation}
for any $x \in \mathbb{R}^{|V|}$. Choosing $x = [1, 1, ...]^T \in \mathbb{R}^{|V|}$ we find
\begin{equation}
    R(A_{G(V)}, x) = \frac{1}{|V|} \sum_{i,j} (A_{G(V)})_{ij} = 2\frac{|E|}{|V|}
    \label{eqSubgraphRayleigh2}
\end{equation}
since $G$ simple implies $G(V)$ simple. Since $G$ is $d$-regular, the number of edges $|E|$ of $G(V)$ can be counted using the edge boundary \eqref{eqEdgeBoundaryDefinition}
\begin{equation}
    2|E| =  d|V| - |\partial V| \ .
    \label{eqSubgraphEdges}
\end{equation}
Combining \eqref{eqSubgraphRayleigh1}, \eqref{eqSubgraphRayleigh2} and \eqref{eqSubgraphEdges} with the definition of the conductance \eqref{eqConductanceDefinition} proves the statement.
\end{proof}

\section{Detailed derivation of correction of the degree bounds using graph symmetries}
\label{secDetailedSymmetryBound}
In order to put the correction of the degree bound in section \ref{subsecCorrection} of the main text onto a solid footing, we here provide a more detailed explanation. Note that while here we discuss the impact of graph symmetries of the induced sub-graph $G(V)$, our analysis in valid for general simple graphs $G$ and their automorphism groups.\\
In addition to Definition \ref{defPermutationMatrix} and Definition \ref{defAutomorphismGroup} we will make use of the concept of equivalence classes of nodes to formalize the notion of two nodes $z, z' \in V$ being connected by a graph symmetry $\Pi \in \mathcal{S}_V$.
\begin{definition} (Relation $\sim$)
    Let $V$ the set of nodes of a graph $G(V)$. The relation $\sim$ between two nodes $z, z' \in V$ is defined as
    \begin{equation}
        z \sim z' \Longleftrightarrow \exists \Pi \in \mathcal{S}_V : \Pi|z\rangle = |z'\rangle \ .
    \end{equation}
    \label{defEquivalenceRelation}
\end{definition}
It can easily verified that $\sim$ is an equivalence relation, i.e. that it is reflexive, symmetric and transitive, since $\mathcal{S}_V$ has a group structure. We find that:
\begin{enumerate}
    \item $\mathbf{1} \in \mathcal{S}_V \Rightarrow$ $\sim$ reflexive
    \item $\Pi^{-1} \in \mathcal{S}_V$ for all $\Pi \in \mathcal{S}_V \Rightarrow$ $\sim$ symmetric
    \item $\Pi_1 \Pi_2 \in \mathcal{S}_V$ for all $\Pi_1, \Pi_2 \in \mathcal{S}_V \Rightarrow$ $\sim$ transitive
\end{enumerate}
As $\sim$ is an equivalence relation, the set of nodes $V$ can be divided into equivalence classes.
\begin{definition} (Equivalence class)
    Given a graph $G(V)$ and the equivalence relation $\sim$, the nodes $z \in V$ are divided into equivalence classes $[ \xi ]$ defined as
    \begin{equation}
    [\xi] := \{ z \in V: z \sim \xi \in V \}
\end{equation}
where $\xi \in V$ is a representative of the equivalence class $[\xi]$.
\label{defEquivalenceClass}
\end{definition}
In the following we will omit the square brackets for clarity and denote the equivalence classes by the Greek letter $\xi$ only.\\
While the vanishing commutator Eq. \eqref{eqGraphCommutation} implies that $A_{G(V)}$ and any $\Pi \in \mathcal{S}_V$ have a common eigenbasis, for $A_{G(V)}$ with degenerate eigenvalues this property is not transitive, i.e. for $\Pi_1, \Pi_2 \in \mathcal{S}_V$ we find $[\Pi_1, A_{G(V)}] = 0$ and $[\Pi_2, A_{G(V)}] = 0$, but that does not necessarily imply $[\Pi_1, \Pi_2] = 0$.\\
However, any eigenvector of $A_{G(V)}$ corresponding to a non-degenerate eigenvalue must also be an eigenvector for all $\Pi \in \mathcal{S}_V$, as mentionend in the main text. We can use this to estimate the principal eigenvalue if $G(V)$ is connected, since then the principal eigenvalue is non-degenerate, according to the Perron-Frobenius theorem.\\
The permutation matrices are orthogonal and, as such, all their eigenvalues are on the complex unit circle
\begin{equation}
    \Pi |x\rangle = \lambda_\Pi |x\rangle \text{  with  } |\lambda_\Pi| = 1
    \label{eqPermutationEigenvalueEquation1}
\end{equation}
for an eigenvector $|x\rangle = \sum_z a_z |z\rangle$. Recall that the permutaion matrices bijectively map the set of node onto itself such that
\begin{equation}
    \Pi |z \rangle = |z'\rangle \ .
\end{equation}
This gives a condition for the coefficients of $|x\rangle$, since
\begin{equation}
    \Pi |x \rangle = \sum_z a_z |z'\rangle = \sum_z \lambda_\Pi a_z |z\rangle\label{eqPermutationEigenvalueEquation2}
\end{equation}
where the last equality follows from Eq. \eqref{eqPermutationEigenvalueEquation1}. As $\Pi$ is a bijection, the sums in Eq. \eqref{eqPermutationEigenvalueEquation2} both run over all $z$. Element-wise equality in Eq. \eqref{eqPermutationEigenvalueEquation2} requires that the coefficients of the eigenstate $|x\rangle$ have identical moduli and fixed phase relations as $\lambda_\Pi$ is on the complex unit circle, i.e.
\begin{equation}
\begin{aligned}
    |a_{z'}| &= |a_z| \ , \\
    \arg(a_{z}) - &\arg(a_{z'}) \equiv \arg(\lambda_\Pi) \mod 2\pi \ .
\end{aligned}
\label{eqAmplitudePhaseRelation}
\end{equation}
This is true for all $\Pi \in \mathcal{S}_V$, hence if $|x\rangle$ is a mutual eigenvector of all $\Pi \in \mathcal{S}_V$, all $a_z$ of nodes within the same equivalence class $\xi$ have the same modulus and fixed phases, up to a global phase factor.\\
Each permutation defines one or more orbits by concatenating the application of the permutation matrix and each orbit has a period $p$ \cite{GarcaPlanas2015}. The period $p$ implies for each $z$ on that orbit that
\begin{equation}
    \Pi^{p} |z\rangle = |z\rangle \ .
\end{equation}
Note that, by definition of $\sim$, all nodes on the same orbit are in the same equivalence class $\xi$.\\
The $\arg (a_z)$ of nodes on the same orbit are distributed equidistantly between 0 and $2\pi$ \cite{GarcaPlanas2015}, unless the eigenvalue $\lambda_\Pi = 1$. It is therefore impossible to write an eigenvector of $\Pi$ with real, positive coefficients $a_z$ up to a global phase factor, if $\lambda_\Pi \neq 1$. But since the principal eigenvector of a connected $G(V)$ can be given with real, positive coefficients (up to a global phase), the principal eigenvector must be an eigenvector of all $\Pi \in \mathcal{S}_V$ with eigenvalue $\lambda_\Pi = 1$. This implies that the principal eigenvector $|x\rangle$ of $G(V)$ is of the form
\begin{equation}
    |x\rangle = \sum_\xi x_\xi |\xi \rangle
\end{equation}
with
\begin{equation}
    |\xi\rangle = \frac{1}{\sqrt{|\xi|}} \sum_{z \in \xi} |z \rangle \ .
\end{equation}
Evaluating the matrix elements
\begin{equation}
    \langle \xi | A_{G(V)} | \xi'\rangle = \frac{1}{\sqrt{|\xi||\xi'|}} \sum_{z \in \xi, z' \in \xi' } (A_{G(V)})_{z,z'}
    \label{eqSymmetricMatrixElement}
\end{equation}
transforms the graph into a smaller weighted graph, potentially with loops. As graph automorphisms by definition maintain neighborhood relations between nodes, each node of equivalence class $\xi$ has in its neighborhood the same number of nodes of equivalence class $\xi'$. Therefore, fixing a $z \in \xi$ and counting its neighbors in the support of $|\xi'\rangle$ results in
\begin{equation}
    \sum_{z' \in \xi'} (A_{G(V)})_{z, z'} = |E_{\xi, \xi'}|
\end{equation}
where $|E_{\xi \xi'}|$ is the number of nodes of equivalence class $\xi'$ in the neighborhood of a node of equivalence class $\xi$. This evaluates the first sum in Eq. \eqref{eqSymmetricMatrixElement}. Evaluating the second sum renders
\begin{equation}
    \langle \xi | A_{G(V)} | \xi'\rangle = \sqrt{ \frac{|\xi|}{|\xi'|}} |E_{\xi \xi'}| =\sqrt{ \frac{|\xi'|}{|\xi|}} |E_{\xi' \xi}|
\end{equation}
where the last equality follows, since either summation in Eq. \eqref{eqSymmetricMatrixElement} can be evaluated first.\\
Applying Gershgorin's circle theorem Eq. \eqref{eqGersh} to this matrix, we find
\begin{equation}
    \lambda_V \leq \max_{\xi} \sum_{\xi'} |\langle \xi | A_{G(V)} | \xi' \rangle |
\end{equation}
which gives
\begin{equation}
    \lambda_V \leq \max_{\xi} \sum_{\xi'} \sqrt{ \frac{|\xi|}{|\xi'|}} |E_{\xi \xi'}|
\label{eqGershgorinSymmetric}
\end{equation}
as discussed in the main text.
For comparison, in the computational basis the Gershgorin bound Eq. \eqref{eqGersh} gives
\begin{equation}
    \lambda_V \leq \max_z \sum_{z'} (A_{G(V)})_{zz'}
\end{equation}
where the sum over $z'$ counts the number of nodes of each equivalence class $\xi'$ in the neighborhood of $z$. Since all nodes $z \in \xi$ have the same number of nodes in their neighborhood, the maximum over $z$ can be replaced by a maximum over $\xi$, while the sum over $z'$ can be replaced by a sum over adjacent equivalence classes, in which case 
\begin{equation}
    \lambda_V \leq \max_{\xi} \sum_{\xi'} |E_{\xi \xi'}|
\label{eqGershgorin}
\end{equation}
as stated in the main text.

\end{document}